\documentclass{article}
\usepackage{amssymb,amsmath,enumerate,graphicx}
\usepackage{verbatim}
\usepackage{color}

\usepackage{tikz}
\usetikzlibrary{snakes}
\tikzstyle{tre}=[circle,draw,minimum size=3.5mm]
\newcommand{\etq}[1]{%
\draw (#1) node {\footnotesize $#1$};
}

\renewcommand{\leq}{\leqslant}
\renewcommand{\geq}{\geqslant}

\newcommand{\NN}{\mathbb{N}}

\newcommand{\qed}{}

\newtheorem{theorem}{Theorem}
\newtheorem{proposition}{Proposition}

\begin{document}

\title{The comparison of tree-sibling time consistent phylogenetic networks is graph isomorphism-complete}

\author{Gabriel Cardona$^1$, Merc\`e Llabr\'es $^1$, Francesc  Rossell\'o $^1$, and Gabriel Valiente $^2$
\\ {\small $^1$Department of Mathematics and Computer Science, University of the Balearic Islands}
\\ {\small 07122 Palma de Mallorca, Spain} 
\\  {\small $^1$Algorithms, Bioinformatics, Complexity and Formal Methods Research
Group}
\\{\small Technical University of Catalonia, 08034 Barcelona, Spain}}

\maketitle

\begin{abstract}
In \cite{cardona.ea:sbtstc:2008} we gave  a metric on the class of semibinary tree-sibling time consistent phylogenetic networks that is computable in polynomial time; in particular, the problem of deciding if two networks of this kind are isomorphic is in P. In this paper, we show that if we remove the semibinarity condition above, then the problem becomes much harder. More precisely, we proof that the isomorphism problem for generic tree-sibling time consistent phylogenetic networks is polynomially equivalent to the graph isomorphism problem. Since the latter is believed to be neither in P nor NP-complete, the chances are that it is impossible to define a metric on the class of all tree-sibling time consistent phylogenetic networks that can be computed in polynomial time.
\end{abstract}

\section{Introduction}	
\label{sec:intro}

After the realization that reticulation processes, like hybridizations, recombinations or lateral gene transfers, have been more relevant  in the evolution of life on Earth than previously thought \cite{doolittle:99}, there has been a growing interest in the development of algorithms for the reconstruction of \emph{phylogenetic networks}: graphical models of evolutionary histories that go beyond phylogenetic trees by including \emph{hybrid nodes} of in-degree greater than one representing reticulation events. As the number of  available such algorithms   increases, the need of methods for the comparison of phylogenetic networks also increases, as they are used, for instance, to assess the reliability and robustness of these algorithms \cite{moret.ea:2004,nakhleh.ea:03psb}.

One of the types of phylogenetic networks for which there exist reconstruction methods  \cite{nakhleh.ea:bioinfo06,nakhleh.ea:bioinfo07}  are 
the \emph{tree-sibling time consistent} networks, \emph{TSTC networks}, for short (see \S \ref{sec:prel} for a formal definition). There have been several attempts to define a metric on the class of all TSTC networks on a given set of taxa \cite{nakhleh:phd04}, and  we have recently given a metric on the class of all \emph{semibinary}  TSTC networks, where all hybrid nodes have in-degree two \cite{cardona.ea:sbtstc:2008}, but none of the metrics for phylogenetic networks computable in polynomial time proposed so far satisfies the separation axiom (distance 0 means isomorphism) for generic TSTC networks: see \cite{cardona.ea:tcbb:comparison.1:2008,cardona.ea:tcbb:comparison.2:2008}. In this paper we show why it should come as no surprise: such a metric would solve in polynomial time the graph isomorphism problem.

The graph isomorphism problem is one of the most important decision problems for which the computational complexity is not known yet \cite{goldberg03,KST93}. It is believed to be neither in P nor NP-complete, and subexponential time solutions for it are known. A problem is said to be \emph{graph isomorphism-complete} when it is polynomially equivalent to the graph isomorphism problem. In this paper we show that, for every set $S$ with more than two elements, the isomorphism problem for TSTC phylogenetic networks with taxa bijectively labeled in $S$ is graph isomorphism-complete.

\section{Preliminaries}

\label{sec:prel}

Let $G=(V,E)$ be a non-empty rooted directed acyclic graph  (a \emph{rDAG}, for short). A node of $G$  is a
\emph{leaf} if it has out-degree $0$,  \emph{internal} if its out-degree is $\geq 1$, of \emph{tree} type if its in-degree is $\leq
1$, of \emph{hybrid}  type if its in-degree is $>
1$, and  \emph{elementary} if it is a tree node of out-degree 1. A node $v$ is a
\emph{child} of another node $u$ (and, hence, $u$ is a \emph{parent} of $v$) if
$(u,v)\in E$. Two nodes $u$ and $v$ are \emph{siblings} of each
other if they share a parent. An arc $(u,v)$ in a rDAG is a \emph{tree arc} when $v$ is a tree node, and a \emph{hybridization arc} when $v$ is a hybrid node. The \emph{height} of a node $v$ is the longest length of a directed path from $v$ to a leaf, and the \emph{depth} of $v$ is the longest length of a directed path from the root to $v$.

Given a finite set $S$ of \emph{labels}, a \emph{$S$-rDAG}  is a rDAG
with  its leaves injectively labelled by
$S$. By an isomorphism of $S$-rDAGs we understand an
isomorphism of directed graphs that  preserves the labelling, that is, that maps each leaf in one network to the leaf with the same label in the other network (in particular, isomorphic $S$-rDAGs must have the same sets of actual leaf labels). In a $S$-rDAG, we shall always identify without any further reference every leaf with its label.

A \emph{phylogenetic network} on a set $S$ of \emph{taxa} is a $S$-rDAG such that:
\begin{itemize}
\item No tree node is elementary.
\item Every hybrid node has out-degree   $1$, and its single  child is a tree node.
\end{itemize}
We will say that a phylogenetic network is \emph{tree-sibling} if every
hybrid node has at least one sibling that is a tree node.

A \emph{temporal assignment} \cite{baroni.ea:sb06} on a network $N=(V,E)$ is mapping
$\tau:V\to\NN$ such that:
\begin{enumerate}[(a)]
\item If $v$ is a hybrid node and $(u,v)\in E$, then $\tau(u)=\tau(v)$.
\item If $v$ is a tree node and $(u,v)\in E$, then $\tau(u)<\tau(v)$.
\end{enumerate}
We will say that a phylogenetic network is \emph{time-consistent} if it admits a
temporal assignment. The following
alternative characterization of time consistency will be used later.
For a proof, see \cite{baroni.ea:sb06,cardona.ea:07a}.

\begin{proposition}
 \label{prop:tc}
Let $N=(V,E)$ be a phylogenetic network, let $E_{H}$ be its set of
hybridization arcs, and let $N^{*}=(V,E^{*})$ be the directed graph  with the same
set $V$ of nodes as $N$ and set of arcs $E^{*}=E\cup \{(v,u)\mid (u,v)\in E_H\}$.
Then, $N$ is time consistent if, and only if, $N^{*}$ does not have
any cycle containing some tree arc of $N$.  \qed
\end{proposition}

The underlying biological motivation for the definitions on phylogenetic networks introduced so far is the following. In a phylogenetic network,  
tree nodes model species (either extant, the leaves, or non-extant, the
internal tree nodes), while hybrid nodes model reticulation events,
where different species interact to create new species, the parents of
the hybrid node being the species involved in this event and its
single child being the resulting species. The tree children of a tree node represent direct descendants through mutation. The first condition in the definition of phylogenetic network says that
every non-extant species is assumed to have at least two different direct descendants, be them by mutation or through some reticulation event. This is a very common restriction in any definition of phylogeny (be it a tree or a network), since  species with only one child cannot be reconstructed from biological data. 

The tree-sibling condition says then that, for every reticulation event, at least one of the 
species involved in it must have some descendant through mutation. This condition was introduced with the name \emph{class I} in L. Nakhleh's PhD Thesis \cite{nakhleh:phd04}, and it has reappeared in several phylogenetic network reconstruction methods  \cite{nakhleh.ea:bioinfo06,nakhleh.ea:bioinfo07}.
As far as the time consistency goes, we understand that
the time assigned to
a node represents the time when the corresponding species existed, or
when the reticulation event took place. The first condition  in time consistency
 means then that the species involved in a reticulation event
must coexist in time in order to interact, while the second condition means that  speciation takes some
amount of time to take place.

\section{Main Results}

It is well known \cite{goldberg03,ZKT85} that the isomorphism problem for rDAGs is graph isomorphism-complete.
It turns out that the isomorphism problem for rDAGs with their leaves bijectively labeled in any given set of labels is also graph isomorphism-complete: since we have not been able to find a proof of this easy result in the literature, we provide one here.

\begin{proposition}
For  every non-empty set $S$ of labels, the isomorphism for $S$-rDAGs is graph isomorphism-complete.
\end{proposition}

\begin{proof}
 Without any loss of generality, we assume that $S=\{1,\ldots,n\}\subseteq \NN$.
 
Let us prove first that the isomorphism of $S$-rDAGs reduces to the isomorphism of rDAGs. For every $S$-rDAG $G$, let $G'$ be the rDAG obtained from $G$ by unlabelling its leaves and then, for each $k=1,\ldots,n$, if $G$ contained a leaf labeled with $k$, then adding to this leaf $k$ tree-children leaves; see Fig.~\ref{fig:red1}. The construction of $G'$ from $G=(V,E)$ adds $O(n^2)\leq O(|V|^2)$ nodes and arcs, and therefore it is polynomial in the size of $G$. And $G$ can be reconstructed from $G'$ by simply replacing, for each $k=1,\ldots,n$,  the node of height 1 with $k$ leaves by a leaf labeled with $k$.
Then, it is straightforward to check that, for every pair of $S$-rDAGs $G_1$ and $G_2$ over $S$, 
$G_1\cong G_2$ as $S$-rDAGs if, and only if, $G_1'\cong G_2'$ as rDAGs.

Let us prove now that  the isomorphism of rDAGs reduces to the isomorphism of $S$-rDAGs. For every rDAG $G$, let $G''$ be the $S$-rDAG obtained from $G$ by adding a new node $a$, arcs from each leaf of $G$ to $a$ and finally adding one  child  leaf  to $a$  labeled $1$; see Fig.~\ref{fig:red2}. The construction of $G''$ from $G=(V,E)$ adds 2 nodes and $O(|V|)$ arcs, and therefore it is polynomial.  And $G$ can be reconstructed from $G''$ by simply removing its leaves and its only height 1 node, $a$, as well as all arcs pointing to $a$ or to the leaves. It is straightforward to check that, for every pair of  rDAGs $G_1$ and $G_2$ over $S$, 
$G_1\cong G_2$ if, and only if, $G_1''\cong G_2''$ as $S$-rDAGs. \qed
\end{proof}

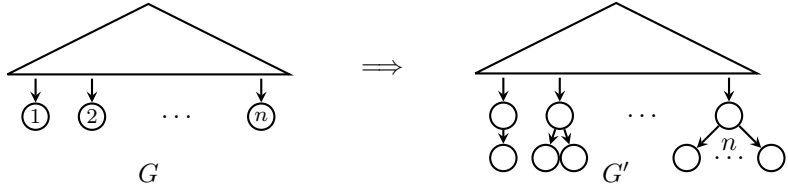
\begin{figure}[htb]
\begin{center}
\begin{tikzpicture}[thick,>=stealth,scale=0.375]
\draw(0,0) node[tre] (1) {}; \etq 1
\draw(2,0) node[tre] (2) {}; \etq 2
\draw(5,0) node  {$\ldots$};  
\draw(8,0) node[tre] (n)  {};  \etq n
\draw (-1,1.5)--(9,1.5)--(4,4)--(-1,1.5);
\draw(0,1.7) node (1u) {};
\draw(2,1.7) node (2u) {};
\draw(8,1.7) node (nu) {};
\draw[->](1u)--(1);
\draw[->](2u)--(2);
\draw[->](nu)--(n);
\draw(4,-2) node {$G$};
\end{tikzpicture}
\qquad
\begin{tikzpicture}[thick,>=stealth,scale=0.375]
\draw(0,-2) node { };
\draw(0,4) node { };
\draw(0,2) node {$\Longrightarrow$};
\end{tikzpicture}
\qquad
\begin{tikzpicture}[thick,>=stealth,scale=0.375]
\draw(0,0) node[tre] (1) {}; 
\draw(2,0) node[tre] (2) {}; 
\draw(5,0) node  {\ldots};  
\draw(8,0) node[tre] (n)  {}; 
\draw (-1,1.5)--(9,1.5)--(4,4)--(-1,1.5);
\draw(0,1.7) node (1u) {};
\draw(2,1.7) node (2u) {};
\draw(8,1.7) node (nu) {};
\draw(0,-1.5) node[tre] (1a) {}; 
\draw(1.5,-1.5) node[tre] (2a) {}; 
\draw(2.5,-1.5) node[tre] (2b) {}; 
\draw(6.5,-1.5) node[tre] (na) {}; 
\draw(9.5,-1.5) node[tre] (nb) {}; 
\draw(8,-1.5) node  {$\ldots$}; 
\draw(8,-1) node  {$n$}; 
\draw[->](1u)--(1);
\draw[->](2u)--(2);
\draw[->](nu)--(n);
\draw[->](1)--(1a);
\draw[->](2)--(2a);
\draw[->](2)--(2b);
\draw[->](n)--(na);
\draw[->](n)--(nb);
\draw(4,-2) node {$G'$};
\end{tikzpicture}
\end{center}
\caption{\label{fig:red1}
The construction involved in the reduction of  the isomorphism of $S$-rDAGs to the isomorphism of rDAGs.}
\end{figure}

\begin{figure}[htb]
\begin{center}
\begin{tikzpicture}[thick,>=stealth,scale=0.375]
\draw(0,0) node[tre] (1) {};  
\draw(2,0) node[tre] (2) {};  
\draw(5,0) node  {$\ldots$};  
\draw(8,0) node[tre] (n)  {}; 
\draw (-1,1.5)--(9,1.5)--(4,4)--(-1,1.5);
\draw(0,1.7) node (1u) {};
\draw(2,1.7) node (2u) {};
\draw(8,1.7) node (nu) {};
\draw[->](1u)--(1);
\draw[->](2u)--(2);
\draw[->](nu)--(n);
\draw(4,-3.5) node {$G$};
\end{tikzpicture}
\qquad
\begin{tikzpicture}[thick,>=stealth,scale=0.375]
\draw(0,-3) node { };
\draw(0,4) node { };
\draw(0,2) node {$\Longrightarrow$};
\end{tikzpicture}
\qquad
\begin{tikzpicture}[thick,>=stealth,scale=0.375]
\draw(0,0) node[tre] (1) {}; 
\draw(2,0) node[tre] (2) {}; 
\draw(5,0) node  {\ldots};  
\draw(8,0) node[tre] (n)  {}; 
\draw (-1,1.5)--(9,1.5)--(4,4)--(-1,1.5);
\draw(0,1.7) node (1u) {};
\draw(2,1.7) node (2u) {};
\draw(8,1.7) node (nu) {};
\draw(4,-1.5) node[tre] (H) {}; 
\draw(4,-3) node[tre] (1a) {}; 
\draw (1a) node {\footnotesize $1$};
\draw[->](1u)--(1);
\draw[->](2u)--(2);
\draw[->](nu)--(n);
\draw[->](1)--(H);
\draw[->](2)--(H);
\draw[->](n)--(H);
\draw[->](H)--(1a);
\draw(2,-3.5) node {$G''$};

\end{tikzpicture}

\end{center}
\caption{\label{fig:red2}
The construction involved in the reduction of  the isomorphism of rDAGs to the isomorphism of $S$-rDAGs.}
\end{figure}
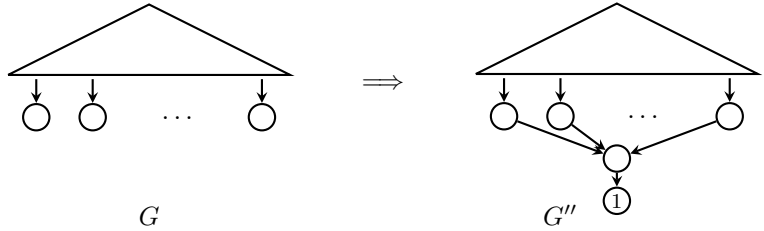

\begin{theorem}
For every set $S$ with $|S|\geq 3$, the isomorphism of TSTC-networks on a set $S$ of taxa is graph isomorphism-complete.
\end{theorem}

\begin{proof}
Without any loss of generality, we assume that $S=\{1,\ldots,n\}\subseteq \NN$.

The isomorphism of TSTC-networks on $S$ clearly reduces to the isomorphism of $S$-rDAGs, since the former are a special case of the latter. Let us prove now the converse reduction.

Let $N=(V,E)$ be a $S$-rDAG.
Let $\overline{N}$ be the $(S\cup\{n+1,n+2\})$-rDAG obtained as follows:
\begin{enumerate}[(1)]
\item For every hybrid node $h$ in $N$, remove all arcs from $h$ to its children, and then add a new (tree) node $u_h$, an arc from $h$ to $u_h$, and new arcs from $u_h$ to the children of $h$ in $N$.
If $h$ was a leaf, say with label $k$, then $u_h$ becomes the new leaf labeled with $k$.

\item For every hybridization arc $e=(v,h)$ in the resulting $S$-rDAG, split it into arcs $(v,v_e)$ and $(v_e,h)$, with $v_e$ a new (tree) node.

Let $N'$ denote the resulting $S$-rDAG after these two first steps.

\item For every internal tree node $v$ in $N'$, add a new (tree) node $v'$ and an arc $(v,v')$.

\item Split the arc $(w,n)$ in $N'$ pointing to the leaf $n$ into two arcs $(w,w_{n})$ and $(w_{n},n)$.

\item Add two new nodes $a$ and $b$, and, for every node $v'$ added in step (3),  add arcs $(v',a)$ and $(v',b)$.  Add also arcs $(w_{n},a)$ and $(w_{n},b)$. The nodes $a$ and $b$ will be hybrid.

\item Add a tree leaf children labelled $n+1$ to $a$, and another one labelled $n+2$ to $b$.
\end{enumerate}
An example of this construction is displayed in Fig.~\ref{fig:red3}.

Let us prove now that $\overline{N}$ is a tree-sibling time consistent phylogenetic network.
\begin{itemize}
\item It is rooted (with the same root as $N$) and acyclic, because all new arcs are either used to split arcs in $N$ into pairs of consecutive arcs, or to define paths that end in the new leaves $n+1$ or $n+2$ without forming cycles.

\item It has no elementary nodes. Indeed, the internal tree nodes in $N$ get an extra child in step (3), 
and the tree nodes that are added to $N$ either get an extra child in step (3) or they get two
children in (5).

\item Its hybrid nodes have only one child, and it is a tree node: this is ensured for the hybrid nodes in $N$ in step (1), and for the new hybrid nodes $a$ and $b$ by construction.

\item It is tree-sibling. All hybrid nodes in $N$ get a tree sibling in steps (2) and (3) (for every hybrid node $h$ in $N$, if $e$ is any arc pointing to $h$, then the tree child $v_e'$ of the new node $v_e$ added in the middle of $e$   is such a tree sibling of $h$), and the hybrid nodes $a$ and $b$ have the tree sibling $n$.

\item It is time consistent. To check this, we use Proposition \ref{prop:tc} (and the notations introduced therein). Since we already know that $\overline{N}$  is acyclic, any cycle in ${\overline{N}}^*$ must contain some inverse of a hybridization arc. 
There are two possibilities for this inverse.
If it has the form $(h,x)$, with $h$ one of the new hybrid nodes $a$ or $b$ introduced in step (5) and $x$ one of the tree nodes $v'$ introduced in step (3) or the tree node $w_{n}$ introduced in step (4), then the only tree arcs that can be reached from $x$ in ${\overline{N}}^*$ are those pointing to the leaves $n$, $n+1$ or $n+2$, and therefore no cycle in ${\overline{N}}^*$ contains this arc $(h,x)$ together with a tree arc.
And if this inverse is of the form $(h,v_e)$, with $h$ a hybrid node in $N$ and $v_e$ one of the tree nodes introduced in step (2), then it must be followed in the cycle by the arc $(v_e,v_e')$ added in step (3), and, as we have just said, the only tree arcs that can be reached from $v_e'$ point to a leaf, and hence  
no cycle in ${\overline{N}}^*$ contains this arc $(h,v_e')$ and a tree arc, either.
\end{itemize}

It is clear that the construction of $\overline{N}$ from $N$ adds $O(|V|+|E|)$ nodes and arcs to $N$, and  
thus it is polynomial in the size of $N$.

Now, the $S$-rDAG $N$ can be easily reproduced from $\overline{N}$ by simply undoing its construction:

\begin{enumerate}
\item[(5)] Remove the leaves $n+1$ and $n+2$ and its hybrid parents $a$ and $b$, together with all arcs pointing to them.

\item[(4)] Remove the elementary parent of the leaf $n$ (which will be the remaining leaf with largest label in $S$)
and replace it by an arc from the parent of the removed node to $n$.

\item[(3)] Remove all non-labeled leaves of the resulting rDAG together with the arcs pointing to them

\item[(2)] Remove each parent $v_e$ of every hybrid node, and replace it by an arc from the parent of $v_e$ to the hybrid child of $v_e$.

\item[(1)] Remove the only tree child of each hybrid node, and replace it by an arc from the hybrid node to each one of the children of the removed node.

\item[(0)] The resulting $S$-rDAG is $N$.
\end{enumerate}
It is straightforward to check now that, for every pair of  $S$-rDAGs $N_1$ and $N_2$, 
$N_1\cong N_2$ if, and only if, $\overline{N_1}\cong \overline{N_2}$ as phylogenetic networks over $S\cup\{n+1,n+2\}$. \qed
\end{proof}

\begin{figure}[htb]
\begin{center}
\begin{tikzpicture}[thick,>=stealth,scale=0.4]
\draw(2,0) node[tre] (1) {}; \etq 1
\draw(0,4) node[tre] (u) {};  
\draw(4,4) node[tre] (v) {};  
\draw[->] (u)--(v);
\draw[->] (u)--(1);
\draw[->] (v)--(1);
\draw (0,-3) node {\ };
\end{tikzpicture}
\begin{tikzpicture}[thick,>=stealth]
\draw(0,0) node { };
\draw(0,4) node { };
\draw(0,2) node {$\stackrel{\mbox{\footnotesize (1)}}{\Longrightarrow}$};
\end{tikzpicture}
\begin{tikzpicture}[thick,>=stealth,scale=0.4]
\draw(2,0) node[tre] (z) {}; 
\draw(2,-2) node[tre] (1) {}; \etq 1
\draw(0,4) node[tre] (u) {};  
\draw(4,4) node[tre] (v) {};  
\draw[->] (u)--(v);
\draw[->] (u)--(z);
\draw[->] (v)--(z);
\draw[->] (z)--(1);
\draw (0,-3) node {\ };
\end{tikzpicture}
\begin{tikzpicture}[thick,>=stealth]
\draw(0,0) node { };
\draw(0,4) node { };
\draw(0,2) node {$\stackrel{\mbox{\footnotesize (2)}}{\Longrightarrow}$};
\end{tikzpicture}
\begin{tikzpicture}[thick,>=stealth,scale=0.4]
\draw(2,0) node[tre] (z) {}; 
\draw(2,-2) node[tre] (1) {}; \etq 1
\draw(0,4) node[tre] (u) {};  
\draw(4,4) node[tre] (v) {};  
\draw(1,2) node[tre] (u1) {};  
\draw(3,2) node[tre] (v1) {};  
\draw[->] (u)--(v);
\draw[->] (u)--(u1);
\draw[->] (v)--(v1);
\draw[->] (u1)--(z);
\draw[->] (v1)--(z);
\draw[->] (z)--(1);
\draw (0,-3) node {\ };
\end{tikzpicture}
\begin{tikzpicture}[thick,>=stealth]
\draw(0,0) node { };
\draw(0,4) node { };
\draw(0,2) node {$\stackrel{\mbox{\footnotesize (3)}}{\Longrightarrow}$};
\end{tikzpicture}
\begin{tikzpicture}[thick,>=stealth,scale=0.4]
\draw(2,0) node[tre] (z) {}; 
\draw(2,-2) node[tre] (1) {}; \etq 1
\draw(0,4) node[tre] (u) {};  
\draw(4,4) node[tre] (v) {};  
\draw(1,2) node[tre] (u1) {};  
\draw(3,2) node[tre] (v1) {};  
\draw(-2,-3) node[tre] (u2) {};  
\draw(0,-3) node[tre] (u12) {};  
\draw(4,-3) node[tre] (v12) {};  
\draw(6,-3) node[tre] (v2) {};  
\draw[->] (u)--(v);
\draw[->] (u)--(u1);
\draw[->] (v)--(v1);
\draw[->] (u1)--(z);
\draw[->] (v1)--(z);
\draw[->] (z)--(1);
\draw[->] (u)--(u2);
\draw[->] (v)--(v2);
\draw[->] (u1)--(u12);
\draw[->] (v1)--(v12);
\end{tikzpicture}
\vspace*{2ex}

\begin{tikzpicture}[thick,>=stealth]
\draw(0,0) node { };
\draw(0,4) node { };
\draw(0,2) node {$\stackrel{\mbox{\footnotesize (4)}}{\Longrightarrow}$};
\end{tikzpicture}
\begin{tikzpicture}[thick,>=stealth,scale=0.4]
\draw(2,0) node[tre] (z) {}; 
\draw(2,-1.5) node[tre] (z1) {}; 
\draw(2,-3) node[tre] (1) {}; \etq 1
\draw(0,4) node[tre] (u) {};  
\draw(4,4) node[tre] (v) {};  
\draw(1,2) node[tre] (u1) {};  
\draw(3,2) node[tre] (v1) {};  
\draw(-2,-3) node[tre] (u2) {};  
\draw(0,-3) node[tre] (u12) {};  
\draw(4,-3) node[tre] (v12) {};  
\draw(6,-3) node[tre] (v2) {};  
\draw[->] (u)--(v);
\draw[->] (u)--(u1);
\draw[->] (v)--(v1);
\draw[->] (u1)--(z);
\draw[->] (v1)--(z);
\draw[->] (z)--(z1);
\draw[->] (z1)--(1);
\draw[->] (u)--(u2);
\draw[->] (v)--(v2);
\draw[->] (u1)--(u12);
\draw[->] (v1)--(v12);
\draw (0,-6) node {\ };
\end{tikzpicture}
\begin{tikzpicture}[thick,>=stealth]
\draw(0,0) node { };
\draw(0,4) node { };
\draw(0,2) node {$\stackrel{\mbox{\footnotesize (5) and (6)}}{\Longrightarrow}$};
\end{tikzpicture}
\begin{tikzpicture}[thick,>=stealth,scale=0.4]
\draw(2,0) node[tre] (z) {}; 
\draw(2,-1.5) node[tre] (z1) {}; 
\draw(2,-3) node[tre] (1) {}; \etq 1
\draw(0,4) node[tre] (u) {};  
\draw(4,4) node[tre] (v) {};  
\draw(1,2) node[tre] (u1) {};  
\draw(3,2) node[tre] (v1) {};  
\draw(-2,-3) node[tre] (u2) {};  
\draw(0,-3) node[tre] (u12) {};  
\draw(4,-3) node[tre] (v12) {};  
\draw(6,-3) node[tre] (v2) {};  
\draw (0,-4.5) node[tre] (A) {};
\draw (4,-4.5) node[tre] (B) {};
\draw (0,-6) node[tre] (2) {}; \etq 2
\draw (4,-6) node[tre] (3) {}; \etq 3
\draw[->] (u)--(v);
\draw[->] (u)--(u1);
\draw[->] (v)--(v1);
\draw[->] (u1)--(z);
\draw[->] (v1)--(z);
\draw[->] (z)--(z1);
\draw[->] (z1)--(1);
\draw[->] (u)--(u2);
\draw[->] (v)--(v2);
\draw[->] (u1)--(u12);
\draw[->] (v1)--(v12);
\draw[->] (u2)--(A);
\draw[->] (v2)--(A);
\draw[->] (z1)--(A);
\draw[->] (u12)--(A);
\draw[->] (v12)--(A);
\draw[->] (u2)--(B);
\draw[->] (v2)--(B);
\draw[->] (z1)--(B);
\draw[->] (u12)--(B);
\draw[->] (v12)--(B);
\draw[->] (A)--(2);
\draw[->] (B)--(3);

\end{tikzpicture}

\end{center}
\caption{\label{fig:red3}
An example of the construction involved in the reduction of  the isomorphism of $S$-rDAGs to the isomorphism of TSTC networks.}
\end{figure}
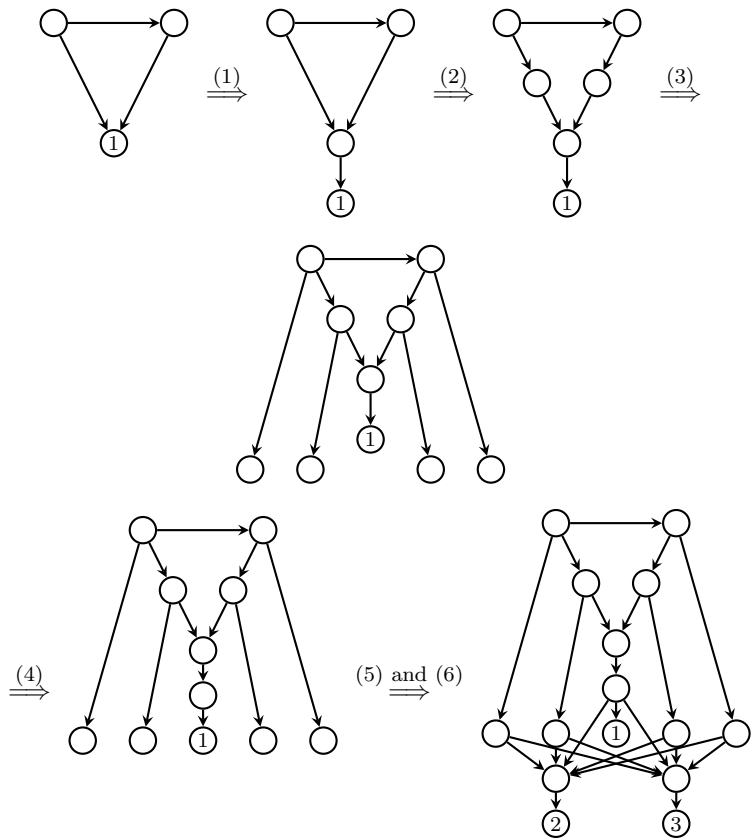

We cannot remove the condition $|S|\geq 3$ in the previous result  because there are only two TSTC phylogenetic networks with less than 3 leaves (up to the actual names of the labels). In particular, this implies that,  in the proof of the previous result,  we cannot  add less than 2 new leaves in the construction of $\overline{N}$ from $N$.

\begin{proposition}
There is only one TSTC phylogenetic network on $\{1\}$, and only one TSTC phylogenetic network on $\{1,2\}$, and in both cases they are trees.
\end{proposition}

\begin{proof}
The $\{1\}$-rDAG consisting of a single node, labeled 1, and the $\{1,2\}$-rDAG consisting of the phylogenetic tree with Newick code \texttt{(1,2);} are clearly TSTC phylogenetic networks. Let us check now that any other TSTC phylogenetic network has at least 3 leaves.

Let $N=(V,E)$ be a TSTC phylogenetic network other than those described in the last paragraph, let 
$\tau: V\to \NN$ be a time assignment, and let $v$ be an internal node with largest $\tau$-value and, among those with this largest time assignment, of largest depth.

If $v$ is a tree node, then all its children are either leaves or
hybrid nodes with leaf children (because any tree descendant node of $v$ has time assignment larger than $\tau(v)$). And $v$'s  hybrid children would have the same time assignment as $v$, but depth largest than $v$'s depth, against the assumption. Therefore all children of $v$ are leaves, and it has at least 2 children, because it cannot be elementary. Now, if $v$ has more than 2 children, we are done, while if it has only two children, say the leaves 1 and 2, then 
$v$ will have a parent in $N$  (because $N$ is not the tree \texttt{(1,2);}). If the parent of $v$ is a tree node, let $w$ be this node, and let $z$ be another child of $w$. Since $N$ does not contain cycles, and any path to $1$ or $2$ must contain $w$, we deduce that any descendant leaf of $z$ must be different from 1 or 2: this gives at least 3 leaves.
If, on the contrary,  the parent of $v$ is a hybrid node $x$, let $w$ be the parent of $x$ that has a tree child, say $z$. The time consistency prevents $x$ to be a descendant of $z$ (because $\tau(z)>\tau(w)=\tau(x)$)  and therefore, since any path leading to 1 or 2 must contain $x$,  any leaf that is a descendant of $z$ will be different from $1,2$: this gives again at least 3 leaves.

If $v$ is a hybrid node, then its child is a leaf, say 1. Let $v_1$ be a parent of $v$ that has a tree child.
Since $\tau(v_1)=\tau(v)$ is the largest $\tau$ value of an internal node of $N$, this tree child must be a leaf, say 2.
Now let $v_2$ be another parent of $v$. Since it is a tree node, it must have another child other than $v$, say $x$. If $x$ is a tree node, it is a leaf, as we have just seen. If $x$ is hybrid, then since $\tau(x)=\tau(v_2)=\tau(v)$, the tree child of $x$ must be a leaf. In both cases, we obtain a leaf that is different from 1 and 2, that is, $N$ contains at least 3 leaves.
\qed
\end{proof}

\section{Conclusion}
We have proved that, unless the graph isomorphism problem belongs to P, there is no hope of defining a polynomially computable metric on the class of all TSTC phylogenetic networks on a set $S$ of at least 3 taxa.  It remains open the problem of defining  polynomially computable, and  biologically sound, metrics on the class of all TSTC phylogenetic networks on a given set $S$ with all their hybrid nodes with in-degree bounded by some $d\in \NN$. When $d=2$, the $\mu$-distance is such a metric \cite{cardona.ea:sbtstc:2008}, but it is no longer a metric for $d=4$ (see 
 the Supplementary Material to the aforementioned paper). Actually, we do not even know whether the isomorphism problem for TSTC phylogenetic networks on a given set $S$ of taxa with 
 globally bounded in-degree hybrid nodes (but without bounding the out-degree of the tree nodes; otherwise, Luks' theorem \cite{luks82} would apply) is in P, but we conjecture that this is the case.

\section*{Acknowledgements}
The authors would like to thank Antoni Lozano for his comments on an early version of this manuscript. 
The research described in this paper has been partially supported by
the  Spanish
DGI project MTM2006-07773 COMGRIO.

\end{document}